\documentclass[11pt,letterpaper]{article}

\usepackage[margin=1in]{geometry}
\usepackage{latexsym,graphicx,amssymb}
\usepackage{amsmath,enumerate}
\usepackage{bbm}
\usepackage{subfigure}
\usepackage{float}
\usepackage{epsfig}
\usepackage{xspace}
\usepackage{paralist}
\usepackage{enumerate}
\usepackage{cases}
\usepackage{caption}
\usepackage{multicol}
\usepackage{graphicx}
\usepackage{xcolor}
\usepackage{tikz}
\usepackage{ifthen}
\usepackage{algorithm}
\usepackage[noend]{algpseudocode}
\usepackage{amsthm}
\usepackage{tcolorbox}
\usepackage{dsfont}
\usepackage{bm}

\newtheorem{theorem}{Theorem}[section]
\newtheorem*{theorem*}{Theorem}

\newtheorem{lemma}{Lemma}[section]
\newtheorem*{lemma*}{Lemma}
\newtheorem{claim}{Claim}[section]
\newtheorem*{claim*}{Claim}

\newtheorem*{fact*}{Fact}

\newtheorem*{remark*}{Remark}

\newcommand{\be}{\begin{equation}}
\newcommand{\ee}{\end{equation}}
\newcommand{\beq}{\begin{equation*}}
\newcommand{\eeq}{\end{equation*}}

\newcommand{\argmax}{\mathop{\rm argmax}}

\newcommand{\R}{\mathbb{R}}

\newcommand{\dd}[1]{\;\mathrm{d}#1}

\newcommand{\AutoAdjust}[3]{\mathchoice{ \left #1 #2  \right #3}{#1 #2 #3}{#1 #2 #3}{#1 #2 #3} }
\newcommand{\Xcomment}[1]{{}}

\newcommand{\InBrackets}[1]{\AutoAdjust{[}{#1}{]}}
\newcommand{\Ex}[2][]{\operatorname{\mathbb E}_{#1}\InBrackets{#2}}
\newcommand{\Exlong}[2][]{\operatornamewithlimits{\mathbb E}\limits_{#1}\InBrackets{#2}}

\newcommand{\eqdef}{\overset{\mathrm{def}}{=\mathrel{\mkern-3mu}=}}
\newcommand{\vect}[1]{\ensuremath{\bm{#1}}}

\newcommand\restr[2]{{
  \left.\kern-\nulldelimiterspace 
  #1 
  \vphantom{\big|} 
  \right|_{#2} 
  }}

\newcommand{\Acts}{\mathcal{A}}

\newcommand{\Outcome}{\mathcal{Y}}

\title{Optimal Robust Contract Design}
\date{}
\author{
Bo Peng \thanks{ITCS, Key Laboratory of Interdisciplinary Research of Computation and Economics, Shanghai University of Finance and Economics, \texttt{ahqspb@163.sufe.edu.cn, tang.zhihao@mail.shufe.edu.cn}}
\and
Zhihao Gavin Tang \footnotemark[1]
}

\begin{document}

\maketitle

\begin{abstract}
We consider the robust contract design problem when the principal only has limited information about the actions the agent can take. The principal evaluates a contract according to its worst-case performance caused by the uncertain action space. Carroll (AER 2015) showed that a linear contract is optimal among deterministic contracts. Recently, Kambhampati (JET 2023) showed that the principal's payoff can be strictly increased via randomization over linear contracts. In this paper, we characterize the optimal randomized contract, which remains linear and admits a closed form of its cumulative density function.
The advantage of randomized contracts over deterministic contracts can be arbitrarily large even when the principal knows only one non-trivial action of the agent. Furthermore, our result generalizes to the model of contracting with teams, by Dai and Toikka (Econometrica 2022).
\end{abstract}

\section{Introduction}
\label{sec:intro}
In contract theory~\cite{holmstrom1979moral, grossman1992analysis}, a principal aims to incentivize an agent to perform a costly and hidden action through monetary reward. Each action results in a stochastic observable outcome. A contract commits to a transfer rule from the principal to the agent that depends on the observable outcome.
Given a contract, the agent takes the action to maximize her utility, the expected wage obtained from the principal minus the cost of the chosen action. The principal aims to maximize his own payoff, that equals the expected outcome minus the wage.

Carroll~\cite{carroll2015robustness} proposed the robust contract design framework in which the principal is uncertain about the actions the agent can take and aims to maximize the worst-case payoff under his limited knowledge. Carroll proved that within the family of deterministic contracts, there exists an optimal linear contract: paying the agent a constant fraction of the outcome. This model justifies the wide application of linear contracts in practice.
Recently, Kambhampati~\cite{kambhampati2023randomization} proved that a randomization over two linear contracts strictly increases the expected payoff of the principal. Although the established advantage of randomized contracts is only $\varepsilon$ large, this result suggests the potential of randomized contracts.

In this paper, we fully characterize the optimal robust randomized contract. We prove that the optimal randomized contract remains linear, and provide its cumulative distribution function in a closed form. The advantage of randomized contracts over deterministic contracts can be arbitrarily large. On the technical side, we formulate the optimal robust contract design problem as a linear program and observe interesting connections between contract design and mechanism design.
Furthermore, we generalize our randomized linear contract to the team setting by Dai and Toikka~\cite{dai2022robust} in which multiple agents are involved.

After the conference submission of our work, we learn an independent and simultaneous work by Kambhampati, Toikka, and Vohra~\cite{kambhampatTV2024}. They achieve the same randomized contract for the single agent case using a different approach. 
Our analysis is arguably simpler and more direct than their approach, and smoothly generalizes to the setting of teams (i.e., multiple agents). To the best of our knowledge, it is unclear how their argument can be applied to the team setting.

The rest of the paper is organized as follows. Section~\ref{sec:model} provides the robust contract design model by Carroll~\cite{carroll2015robustness}. 
Our main result is established through a ``guess and verify'' approach. 
The guess part is presented in Section~\ref{sec:lp}, which formulates the robust contract design problem as a linear program. 
The verification part is presented in Section~\ref{sec:single-agent}.
Finally, the extension to the team setting is provided in Section~\ref{sec:team}.

\section{Model}
\label{sec:model}
We state the robust contract design framework for a single agent by Carroll~\cite{carroll2015robustness}. The team model shall be provided in Section~\ref{sec:team}.
\paragraph{Robust Contract Design.} A principal contracts with an agent who takes a costly action $a \in \Acts$ that leads to a stochastic outcome $y \in \mathcal{Y}$. 
Only the outcome $y$ but not the action $a$ is observable by the principal. Hence, the payment provided by the principal to the agent can only depend on $y$.
We assume that the outcome is an element of $\mathcal{Y}$, a compact subset of $\R$. A contract is then defined as a function $w: \mathcal{Y} \to \mathbb{R}_+$.

The technology (i.e., the action set) $\Acts$ is a subset of $\Delta(\mathcal{Y}) \times \mathbb{R}_{+}$. That is, each action is a pair $(F,c)$ meaning that the agent pays $c$ and the outcome is drawn from $F$. Given an arbitrary (deterministic) contract $w$, the utility of the agent by taking an action $a=(F,c)$ is 
\[
u(w,a) \eqdef \Exlong[F]{w(y)} - c~.
\]
We assume that the agent always takes the utility-maximizing action. 

The principal only knows a subset $\Acts_0$ of the actual technology $\Acts$ and designs a (randomized) contract to maximize his worst-case payoff:
\[
\max_{w \sim G} \min_{\Acts \supseteq \Acts_0} \left( \Exlong[w\sim G]{\Exlong[y \sim F(w)]{y-w(y)}} \right),
\]
where $a(w) = (F(w),c(w)) \in \Acts$ is a utility-maximizing action of the agent with respect to $w$. 

\paragraph{Tie-breaking.} Note that we need to specify the tie-breaking rule when the agent is indifferent among multiple actions, so that the above payoff is well-defined. Previous works~\cite{carroll2015robustness, kambhampati2023randomization} assume that the agent breaks tie in favor of the principal, i.e., she chooses the best action for the principal. Another important and natural tie-breaking rule is the opposite, in which the agent chooses the worst action for the principal. We refer to the two tie-breaking rules as best and worst tie -breaking respectively and define the following two values of the principal:
\begin{align*}
\text{Best tie-breaking:} \quad & \overline{V_P}(\Acts_0) \eqdef \max_{w \sim G} \min_{\Acts \supseteq \Acts_0} \left( \max_{a(w) \in \argmax_{a\in \Acts} u(w,a)} \left( \Exlong[w\sim G]{\Exlong[y \sim F(w)]{y-w(y)}} \right) \right) \\
\text{Worst tie-breaking:} \quad & \underline{V_P}(\Acts_0) \eqdef \max_{w \sim G} \min_{\Acts \supseteq \Acts_0} \left( \min_{a(w) \in \argmax_{a\in \Acts} u(w,a)} \left( \Exlong[w\sim G]{\Exlong[y \sim F(w)]{y-w(y)}} \right) \right)
\end{align*}
Straightforwardly, $\overline{V_P}(\Acts_0) \ge \underline{V_P}(\Acts_0)$. Our main theorem states that the two values are equal (in the degenerated case) by establishing an upper bound on $\overline{V_P}(\Acts_0)$ and an lower bound on $\underline{V_P}(\Acts_0)$. More details are discussed in Section~\ref{sec:single-agent}.
As an implication, our result suggests that the tie-breaking rule is not important for robust contract design in the general case.

\paragraph{Non-triviality Assumption.} We assume that the known technology $\Acts_0$ includes a null action $(\delta_0, 0)$, representing that the agent can exert no effort, where $\delta_0$ denotes the deterministic distribution of value $0$; and includes at least one \emph{non-trivial} action $(F_0,c_0)$ with $\Ex[F_0]{y} - c_0 > 0$.

Finally, we summarize the timeline of the (randomized) contract design problem:
\begin{enumerate}
    \item The principal commits to a randomized contract with distribution $G$, based on the known technology $\Acts_0$.
    \item The adversary learns $G$ and chooses the technology $\Acts \supseteq \Acts_0$.
    \item A contract $w$ is realized according to $G$ and the agent chooses her utility-maximizing action $a=(F,c)$ from $\Acts$ with respect to $w$.
    \item The outcome $y$ is realized according to $F$ and the principal collects $y-w(y)$.
\end{enumerate}
Note that we assume the adversary (the technology $\Acts$) is non-adaptive to the realization of the randomized contract.
In contrast, an adaptive adversary is allowed to choose the technology $\Acts$ after the realization of $w$. It is straightforward to check that randomized contracts cannot improve over deterministic ones with respect to an adaptive adversary.

\section{Linear Program for Robust Contract Design}
\label{sec:lp}
Our analysis differs from the previous works, but is similar to the alternative approach (Appendix C) of Carroll~\cite{carroll2015robustness}, and the ``unsuccessful'' minimax approach (Appendix A.4) as discussed by Kambhampati~\cite{kambhampati2023randomization}.
We solve the optimal contract design problem by formulating it as a \emph{linear} optimization and observe an interesting connection between contract design and mechanism design. 

Though the materials of this section convey the most important ideas of our approach, the formal proofs in Section~\ref{sec:single-agent} are written self-contained and do not necessarily go through the analysis in this section. We intend to keep the formal proofs as clean as possible, so that it is easy to digest by readers who are not familiar with mechanism design.

Within this section, for the ease of presentation, we assume that 1) the outcome space $\Outcome$ is finite, and 2) the set of all possible payment values is finite, i.e., $w(y)$ is an element of a finite set $S$, for every $y \in \Outcome$. Readers might think of that $S$ is a proper discretization of $\R_+$, e.g., integers in a bounded domain.
Under these assumptions, there are only finite number of different contracts where each contract $w$ corresponds to a vector in $S^\Outcome$, i.e., $w_y = w(y)$.

The principal first chooses a distribution $p \in \Delta(S^{\Outcome})$. The adversary then specifies the technology $\Acts \supseteq \Acts_0$. The second step is equivalent to choose an action $a(w)$ for each contract $w \in S^{\Outcome}$:
\[
a(w) = (\vect{q}(w), c(w)) \in \Delta(\Outcome) \times \R_+,
\]
where $\vect{q}(w)$ is a discrete probability distribution supported on $\Outcome$ (i.e., the outcome equals $y$ with probability $q_y(w)$ for contract $w$), and $c(w)$ is the corresponding cost.
Recall that the action $a(w)$ is a utility-maximizing action for the agent. The actions must satisfy the following incentive compatibility constraints:
\[
u(w, a(w)) \ge u(w, a(w')) \iff \sum_{y \in \Outcome} \left( w_y \cdot q_y(w) \right) - c(w) \ge \sum_{y \in \Outcome} \left( w_y \cdot q_y(w') \right) - c(w'), \quad \forall w,w'
\]
The known technology $\Acts_0$ results in the following boundary constraints that subsume the individual rationality constraints:
\[
u(w, a(w)) \ge \max_{a \in \Acts_0} u(w, a) \iff \sum_{y \in \Outcome} \left( w_y \cdot q_y(w) \right) - c(w) \ge \max_{(\vect{q}, c) \in \Acts_0} \left( \sum_{y \in \Outcome} \left( w_y \cdot q_y \right) - c \right)
\]
We denote the right hand side of the last inequality as a function $\underline{u}: S^\Outcome \to \R_+$:
\[
\underline{u}(w) \eqdef \max_{(\vect{q}, c) \in \Acts_0} \left( \sum_{y \in \Outcome} \left( w_y \cdot q_y \right) - c \right)~.
\]
It serves as a lower bound on the utility of the agent with respect to contract $w$.
To sum up, the robust contract design problem can be formulated as the following max-min optimization:
\begin{align*}
\label{lp:max-min}
\max_{p} \min_{\vect{q},c} : \quad & \sum_{w} p(w) \cdot \left( \sum_{y} (y-w_y) \cdot q_y(w) \right) \tag{Max-Min} \\
\text{subject to} : \quad & \sum_{y \in \Outcome} \left( w_y \cdot q_y(w) \right) - c(w) \ge \sum_{y \in \Outcome} \left( w_y \cdot q_y(w') \right) - c(w'), & \forall w, w' \\
& \sum_{y \in \Outcome} \left( w_y \cdot q_y(w) \right) - c(w) \ge \underline{u}(w), & \forall w \\
& \sum_{y} q_y(w) = 1, & \forall w
\end{align*}

\paragraph{Mechanism design counterpart.}
Readers might have found the first and second family of constraints familiar. Indeed, these are standard incentive compatibility constraints and individual rationality constraints in mechanism design. Consider a monopolist selling a set of goods $\Outcome$ to an additive buyer\footnote{We should think of $\Outcome$ as an abstract set of goods.} . The buyer's valuation of the goods can be represented by a vector $w \in S^\Outcome$.  The monopolist designs a menu $\Acts$ with a restriction that $\Acts_0$ must be part of the menu. Each menu item is consisted of an allocation $\vect{q}$ of the goods and a payment $c$. Let $\vect{q}(\cdot), c(\cdot)$ be the corresponding allocation rule and payment rule of menu $\Acts$. Then they should satisfy the same IC and IR constraints as above. A usual mechanism design problem aims to maximize a certain objective (e.g., the social welfare or the revenue of the monopolist), while the contract design problem is a minimization problem. To the best of our knowledge, this connection between contract design and mechanism design in unknown prior to our work. This intriguing observation allows us to apply tools from mechanism design to solve our contract design problem. 

On the other hand, it is far from obvious that the optimal solution of \eqref{lp:max-min} is a (randomized) linear contract, as the mechanism design counterpart is intrinsically multi-dimensional, which is known to be difficult and often has a complicated optimal solution.
Fortunately, as we shall see in the next section, the optimization of contract design degenerates to the single-dimensional case and we enjoy benefits from the celebrated Myerson's Lemma in mechanism design. 

Coming back to the max-min optimization, notice that the inner optimization is linear in $\vect{q}, c$. We can write the corresponding dual program and claim strong duality:
\begin{align*}
\max_{\lambda, \mu, \theta}: \quad & \sum_{w} \mu(w) \cdot \underline{u}(w) + \sum_{w} \theta(w) \\
\text{subject to:} \quad & \sum_{w'} \left( w_y \cdot \lambda(w, w') - w'_y \cdot \lambda(w', w) \right) + w_y \cdot \mu(w) + \theta(w) \le p(w) \cdot (y-w_y), & \forall w, y \\
& \sum_{w'} \left( \lambda(w', w) - \lambda(w, w')\right) \le \mu(w), & \forall w
\end{align*}
Here, the variables $\lambda(w,w')$ corresponds to the IC constraints; $\mu(w)$ corresponds to the IR constraints; and $\theta(w)$ corresponds to the third family of constraints. This is similar to the duality framework of Cai, Devanur, and Weinberg~\cite{sicomp/CaiDW21}, and the partial Lagrangian analysis of Carroll~\cite{carroll2017robustness}, developed for multi-dimensional mechanism design.

Nevertheless, the exact form of the above optimization is not important. A crucial observation is that the original max-min problem is now transferred to a max-max problem through strong duality, and the max-max optimization is linear as a whole. To sum up, we reach the linear program for optimal contract design:
\begin{align*}
\label{lp:max-max}
\max_{p, \lambda, \mu, \theta}: \quad & \sum_{w} \mu(w) \cdot \underline{u}(w) + \sum_{w} \theta(w) \tag{Max-Max} \\
\text{subject to:} \quad & \sum_{w'} \left( w_y \cdot \lambda(w, w') - w'_y \cdot \lambda(w', w) \right) + w_y \cdot \mu(w) + \theta(w) \le p(w) \cdot (y-w_y), & \forall w, y \\
& \sum_{w'} \left( \lambda(w', w) - \lambda(w, w')\right) \le \mu(w), & \forall w \\
& \sum_{w} p(w) = 1
\end{align*}
Our main result indeed solves this linear program in a closed analytic form, but in an implicit way as we shall use a degenerated single-dimensional linear program\footnote{We remark that we implement this complicated linear program at an early stage of the project that leads us to the convincing conjecture that the optimal randomized contract is linear.}. Moreover, with a slight modification to the above approach, we provide an alternative proof of the linearity of the optimal deterministic contract, i.e., the main theorem of Carroll~\cite{carroll2015robustness}. And the approach smoothly generalizes to the multi-observable outcome setting, for which Carroll believed his alternative approach is difficult to apply. The proofs are provided in Appendix~\ref{app:deterministic}. 

A final remark is that the idea of transforming a \emph{multi-linear} max-min optimization (refer to \eqref{lp:max-min}) into a \emph{linear} max-max optimization (refer to \eqref{lp:max-max}) has appeared in the literature of mechanism design. Specifically, Gravin and Lu~\cite{soda/GravinL18} and Bei et al.~\cite{soda/BeiGLT19} established similar linear programs to characterize correlation-robust mechanism design problems. 
\section{Optimal Randomized Contracts}
\label{sec:single-agent}
In this section, we prove our main theorem that randomized linear contracts emerge as the optimal robust contract. Moreover, the cumulative distribution function of the optimal randomized linear contract admits an analytic form.

From now on, we fix an arbitrary finite technology set $\Acts_0$ with at least one non-trivial action. 
We define the following auxiliary function:
\[
\underline{u}(\alpha) \eqdef \max_{(F_0,c_0) \in \Acts_0} \left( \alpha \cdot \Exlong[F_0]{y} - c_0 \right)~.
\]
This function is a utility lower bound of the agent with respect to the linear contract $w_\alpha(y) \eqdef \alpha \cdot y$, since the actual technology $\Acts$ is a superset of $\Acts_0$ and the agent is assumed to maximizes her own utility. 
Observe that $\underline{u}(\cdot)$ is a continuous non-decreasing function with $\underline{u}(0)=0$ and $0<\underline{u}(1)< \infty$. 
Further, $\underline{u}(\cdot)$ is convex and piece-wise linear. But we are not going to use these properties.

Let $\alpha^*$ denote the critical slope that
\[
\alpha^* \in \argmax_{\alpha \in [0,1]} \frac{\underline{u}(\alpha)}{- \ln (1-\alpha)}~.
\]
Consider the following cumulative distribution function:
\[
G^*(\alpha) \eqdef \frac{\ln{(1-\alpha)}}{\ln{(1-\alpha^*)}}, \quad \forall \alpha \in [0,\alpha^*]~.
\]

\begin{theorem}
\label{thm:main}
The optimal randomized contract achieves an expected payoff of $\frac{\underline{u}(\alpha^*)}{- \ln (1-\alpha^*)}~.$
And it is achieved by a randomized linear contract $w_\alpha(\cdot)$, where $\alpha$ is drawn according to the cumulative distribution function $G^*$.
\end{theorem}

First, we establish a lower bound on the expected payoff of our randomized contract. This involves a subtle tie-breaking issue depending on whether $\alpha^*=0$. Notice that our randomized contract degenerates to a deterministic one when $\alpha^*=0$ (i.e., $G^*(0)=1$). Nevertheless, we show that the claimed expected payoff coincides with the optimal deterministic payoff of Carroll~\cite{carroll2015robustness} when the agent breaks tie in favor of the principle. 
\begin{lemma}
\label{lem:degenerate}
    If $\alpha^*=0$, then $\overline{V_P}(\Acts_0) \ge \max_{(F_0,c_0) \in \Acts_0} \left( \sqrt{\Ex[F_0]{y}} - \sqrt{c_0}\right)^2 = \frac{\underline{u}(\alpha^*)}{- \ln (1-\alpha^*)}$, which is achieved by the zero-slope linear contract $w_0(\cdot)$.
\end{lemma}
In the non-degenerate case when $\alpha^* > 0$, we establish the stated payoff even when the agent breaks tie in the worst case. 
\begin{lemma}
\label{lem:lower}
If $\alpha^* > 0$, $\underline{V_P}(\Acts_0) \ge \frac{\underline{u}(\alpha^*)}{- \ln (1-\alpha^*)}$. Specifically, the randomized linear contract $w_\alpha(\cdot)$ achieves an expected payoff of $\frac{\underline{u}(\alpha^*)}{- \ln (1-\alpha^*)}$, where $\alpha$ is drawn according to the cumulative distribution function $G^*$.
\end{lemma}

Second, we prove that no contract can achieve a better payoff with respect to an explicitly constructed technology $\Acts \supseteq \Acts_0$. Our main theorem is then a consequence of the two steps.
\begin{lemma}
\label{lem:upper}
$\overline{V_P}(\Acts_0) \le \frac{\underline{u}(\alpha^*)}{- \ln (1-\alpha^*)}$. Specifically, there exists a technology $\Acts \supseteq \Acts_0$ such that no contract achieves a payoff larger than $\frac{\underline{u}(\alpha^*)}{- \ln (1-\alpha^*)}$ with respect to $\Acts$.
\end{lemma}

\subsection{Proof of Lemma~\ref{lem:degenerate}}
We first prove the equality $\max_{(F_0,c_0) \in \Acts_0} \left( \sqrt{\Ex[F_0]{y}} - \sqrt{c_0}\right)^2 = \frac{\underline{u}(\alpha^*)}{- \ln (1-\alpha^*)}$ as stated in the lemma. 
Since $\alpha^*=0$, there must exist an action $(F^{*},c^{*}) \in \Acts_0$ with $c^*=0$ and
\begin{equation}
\label{eqn:h}
\Exlong[F^*]{y} = \lim_{\alpha \to 0}\frac{\underline{u}(\alpha)}{- \ln (1-\alpha)} \ge \frac{\underline{u}(\alpha)}{- \ln (1-\alpha)}, \quad  \forall \alpha \in [0,1]~.
\end{equation}
Hence, it suffices to prove that
\begin{equation}
\label{eqn:degenerate}
\Exlong[F^*]{y} \ge \left( \sqrt{\Exlong[F_0]{y}} - \sqrt{c_0}\right)^2, \quad \forall (F_0,c_0) \in \Acts_0~.
\end{equation}
Fix an arbitrary $(F_0,c_0) \in \Acts_0$.
If $\Ex[F^{*}]{y} \ge \Ex[F_0]{y}$, we have $ \sqrt{\Ex[F^{*}]{y}} \ge \sqrt{\Ex[F_0]{y}} \ge \sqrt{\Ex[F_0]{y}}-\sqrt{c_0}$.
Else, let $h(\alpha) \eqdef \Ex[F^*]{y} \cdot (- \ln (1-\alpha)) - \Ex[F_0]{y} \cdot \alpha + c_0$. By \eqref{eqn:h}, we have $h(\alpha) \ge 0$ for every $\alpha \in [0,1]$.
Note that $h$ attains its minimum at $\alpha=1-\frac{\Ex[F^{*}]{y}}{\Ex[F_0]{y}}$. Thus, we have
\begin{multline*}
h\left(1-\frac{\Ex[F^{*}]{y}}{\Ex[F_0]{y}}\right) =  \Exlong[F^{*}]{y} \cdot \left( -\ln\left( \frac{\Ex[F^{*}]{y}}{\Ex[F_0]{y}}\right)\right) -\Exlong[F_0]{y} +\Exlong[F^{*}]{y} +c_0 \ge 0 \\
\Longrightarrow 
\sqrt{\Exlong[F_0]{y}}-\sqrt{c_0} \le \sqrt{\Exlong[F_0]{y}} - \sqrt{\Exlong[F_0]{y} -\Exlong[F^{*}]{y} +\Exlong[F^{*}]{y} \cdot \ln\left( \frac{\Ex[F^{*}]{y}}{\Ex[F_0]{y}}\right)} \\
= \sqrt{\Exlong[F_0]{y}} \cdot \left(1-\sqrt{1-\frac{\Ex[F^{*}]{y}}{\Ex[F_0]{y}}+\frac{\Ex[F^{*}]{y}}{\Ex[F_0]{y}} \cdot \ln\left(\frac{\Ex[F^{*}]{y}}{\Ex[F_0]{y}}\right)}\right) \le \sqrt{\Exlong[F_0]{y}} \cdot \sqrt{\frac{\Ex[F^{*}]{y}}{\Ex[F_0]{y}}} = \sqrt{\Exlong[F^*]{y}},
\end{multline*}
where the last inequality follows from the mathematical fact that $1-\sqrt{1-x+x\cdot \ln x} \le \sqrt{x}$ for $x \in [0,1]$. This finishes the proof of equation~\eqref{eqn:degenerate}.
In other words, $(F^*,c^*)$ is the maximizer of 
\[
\max_{(F_0,c_0) \in \Acts_0} \left( \sqrt{\Exlong[F_0]{y}} - \sqrt{c_0}\right)^2~.
\]
Together with the main theorem of Carroll~\cite{carroll2015robustness} finishes the proof of our lemma.
\begin{theorem*}[Carroll~\cite{carroll2015robustness}]
The optimal deterministic contract achieves a payoff of 
\[
\max_{(F_0,c_0)\in \Acts_0} \left( \sqrt{\Exlong[F_0]{y}} - \sqrt{c_0}\right)^2
\]
when the agent breaks tie in favor of the principle. Moreover, it is achieved by a linear contract with slope $\sqrt{\frac{c^*}{{\Ex[F^*]{y}}}}$, where
\[
(F^*,c^*) \in \argmax_{(F_0,c_0) \in \Acts_0} \left( \sqrt{\Ex[F_0]{y}} - \sqrt{c_0}\right)^2~.
\]
\end{theorem*} 
\subsection{Proof of Lemma~\ref{lem:lower}}

We establish a linear program that lower bounds the expected payoff of any randomized \emph{linear} contract.
\begin{claim}
\label{clm:lp1}
For an arbitrary randomized linear contract with cumulative distribution function $G$ (defined on the slope), the expected payoff of the principal is lower bounded by the following program.
\begin{align}
\label{eqn:lp1}
\min_{e, c}: \quad & \int_0^1 (1-\alpha) \cdot e(\alpha) \dd G(\alpha) \tag{P1}\\
{\normalfont \text{subject to}}: \quad & \alpha \cdot e(\alpha) - c(\alpha) \ge \alpha \cdot e(\alpha') - c(\alpha') & \forall \alpha, \alpha' \in [0,1] \notag \\
& \alpha \cdot e(\alpha) - c(\alpha) \ge \underline{u}(\alpha) & \forall \alpha \in [0,1] \notag \\
& e(\alpha), c(\alpha) \ge 0 & \forall \alpha \in [0,1] \notag
\end{align}
\end{claim}
\begin{proof}
For an arbitrary set technology $\Acts \supseteq \Acts_0$, let $a(\alpha) = (F(\alpha), c(\alpha))$ be the action that maximizes her own utility when the linear contract uses a slope of $\alpha$, i.e.,
\[
a(\alpha) = (F(\alpha),c(\alpha)) \eqdef \argmax_{(F,c) \in \Acts} \left( \alpha \cdot \Exlong[F]{y} - c \right)~.
\]
Here, we break tie in the worst case for the principal.
We abuse the notation $u(\alpha)$ to denote the utility of the agent with respect to the linear contract $w_{\alpha}$, i.e.,
\[
u(\alpha) \eqdef \max_{(F,c) \in \Acts} \left( \alpha \cdot \Exlong[F]{y} - c \right) = \alpha \cdot \Exlong[F(\alpha)]{y} - c(\alpha)~.
\]
We use $e(\alpha) \eqdef \Ex[F(\alpha)]{y}$ to denote the expected outcome of $F(\alpha)$. Then, $u(\alpha) = \alpha \cdot e(\alpha) - c(\alpha)$ and the corresponding payoff of the principal is $(1-\alpha) \cdot e(\alpha)$.
The first family of constraints follows from the fact that action $a(\alpha)$ provides a larger utility than $a(\alpha')$ when the contract is realized to be $w_\alpha$. That is,
\[
u(w_{\alpha}, a(\alpha)) \ge u(w_{\alpha}, a(\alpha')) \iff \alpha \cdot e(\alpha) - c(\alpha) \ge \alpha \cdot e(\alpha') - c(\alpha')~.
\]
Furthermore, $a(\alpha)$ also guarantees a no smaller utility than every action $(F_0,c_0) \in \Acts_0$. This gives the second family of the constraints:
\[
u(w_{\alpha}, a(\alpha)) \ge \max_{(F_0,c_0) \in \Acts_0} u(w_{\alpha}, (F_0,c_0)) \iff \alpha \cdot e(\alpha) - c(\alpha) \ge \underline{u}(\alpha)~.
\]
This concludes the proof of the statement.
\end{proof}

Readers who are familiar with mechanism designs might have noticed that the first family of constraints of \eqref{eqn:lp1} are essentially the incentive compatibility constraints in a single-dimensional auction setting. 
Our functions $e(\cdot), c(\cdot)$ correspond to the allocation rule and payment rule respectively.
The celebrated Myerson's lemma~\cite{mor/Myerson81} states that the function $e(\alpha)$ must be non-decreasing in $\alpha$ and the function $c(\alpha)$ is uniquely determined by $e(\alpha)$. This allows us to simplify \eqref{eqn:lp1}.
For completeness, we provide a self-contained proof in our context, which is a verbatim proof from Myerson's Lemma.

\begin{claim}
\label{clm:lp2}
Program \eqref{eqn:lp1} is equivalent to the following program:
\begin{align}
\label{eqn:lp2}
\min_{e}: \quad & \int_0^1 (1-\alpha) \cdot e(\alpha) \dd G(\alpha) \tag{P2} \\
{\normalfont \text{subject to}}: \quad & e(\alpha) \ge e(\alpha') \ge 0 & \forall 0 \le \alpha' < \alpha \le 1 \notag \\
& \int_0^{\alpha} e(t) \dd t \ge \underline{u}(\alpha) & \forall \alpha \in [0,1] \notag 
\end{align}
\end{claim}

\begin{proof}
By the first family of constraints in \eqref{eqn:lp1}, for every $\alpha > \alpha'$, we obtain
\[
\alpha \cdot (e(\alpha)-e(\alpha')) \ge c(\alpha)-c(\alpha') \ge \alpha' \cdot (e(\alpha)-e(\alpha')),
\]
which implies that $e(\cdot)$ is non-decreasing. Furthermore, for any non-decreasing $e(\cdot)$, we have that
\[
\alpha \cdot \frac{e(\alpha)-e(\alpha')}{\alpha- \alpha'}  \ge \frac{c(\alpha)-c(\alpha')}{\alpha-\alpha'} \ge \alpha' \cdot \frac{e(\alpha)-e(\alpha')}{\alpha- \alpha'} 
\]
Thus, $\frac{\dd c(\alpha)}{\dd \alpha} = \alpha \cdot \frac{\dd e(\alpha)}{\dd \alpha}$ by letting $\alpha' \to \alpha-$0\footnote{For simplicity, we assume $e(\cdot)$ to be differentiable but the statement generalizes to an arbitrary non-decreasing function $e(\cdot)$. See~\cite{mor/Myerson81} for a more formal treatment.}. Integrating $\alpha$ from $0$, we have that 
\[
c(\alpha) - c(0) = \alpha \cdot e(\alpha) - \int_0^{\alpha} e(t) \dd t~.
\]
Finally, by the non-negativity of $c(\cdot)$, a dominating choice of $c(0)$ is $0$ so that the IR constraints of \eqref{eqn:lp1} are satisfied. This concludes the proof of the claim.
\end{proof}

Finally, by Claim~\ref{clm:lp1} and \ref{clm:lp2}, the expected payoff of the randomized linear contract with cumulative distribution function $G^*$ is lower bounded by
\begin{align*}
\min_{e}: \quad & \int_0^1 (1-\alpha) \cdot e(\alpha) \dd G^*(\alpha) \\
{\normalfont \text{subject to}}: \quad & e(\alpha) \ge e(\alpha') \ge 0 & \forall 0 \le \alpha' < \alpha \le 1 \\
& \int_0^{\alpha} e(t) \dd t \ge \underline{u}(\alpha) & \forall \alpha \in [0,1]
\end{align*}
According to the definition of $G^*$, we have
\begin{align*}
\int_0^1 (1-\alpha) \cdot e(\alpha) \dd G^*(\alpha) & = \int_0^{\alpha^*} (1-\alpha) \cdot e(\alpha) \dd \left( \frac{\ln(1-\alpha)}{\ln(1-\alpha^*)} \right) \\
& = \int_0^{\alpha^*}  \frac{e(\alpha)}{-\ln{(1-\alpha^*)}} \dd \alpha = \frac{\int_0^{\alpha^*}  e(\alpha) \dd \alpha}{-\ln{(1-\alpha^*)}} \ge \frac{\underline{u}(\alpha^*)}{-\ln{(1-\alpha^*)}}~,
\end{align*}
which concludes the proof of the lemma.

\begin{remark*}
Our distribution $G^*$ is indeed the optimal solution of the robust optimization:
\[
\max_{G} \min_{e} \int_0^1 (1-\alpha) \cdot e(\alpha) \dd G(\alpha),
\]
where the feasible spaces of $G$ and $e$ are both convex ($G$ is a cumulative distribution function and $\vect{e}$ is a function that satisfies the (linear) constraints as stated in \eqref{eqn:lp2}). By the minimax theorem, we have that
\[
\max_{G} \min_{e} \int_0^1 (1-\alpha) \cdot e(\alpha) \dd G(\alpha) = \min_{e} \max_{G} \int_0^1 (1-\alpha) \cdot e(\alpha) \dd G(\alpha)~.
\]
We shall utilize the optimal solution $e^*$ of the min-max optimization to prove Lemma~\ref{lem:upper}.
\end{remark*}

\subsection{Proof of Lemma~\ref{lem:upper}}
Let $e^*$ be the following function:
\[
e^*(\alpha) \eqdef \frac{\underline{u}(\alpha^*)}{- (1-\alpha) \cdot \ln{(1-\alpha^*)}}, \quad \forall \alpha \in [0,1]~,
\]
and $c^*$ be the corresponding cost function:
\[
c^*(\alpha) \eqdef \alpha \cdot e^*(\alpha) - \int_0^{\alpha} e^*(t) \dd t, \quad \forall \alpha \in [0,1]~.
\]
Let $\bar{\alpha}$ be the value that $e^*(\bar{\alpha}) = \max_{(F_0,c_0) \in \Acts_0} \Ex[F_0]{y}$. Note that such an $\bar{\alpha}$ exists since $\lim_{\alpha \to 1} e^*(\alpha) = +\infty$.

Then for every $\alpha \in [0,\bar{\alpha}]$, define an action set:
\begin{align*}
& \Acts(0) \eqdef \left\{ (F,c) \middle| \Exlong[F]{y} \le e^*(0), c \ge 0 \right\} \\
& \Acts(\alpha) \eqdef \left\{ (F,c) \middle| \Exlong[F]{y}=e^*(\alpha),c \ge c^*(\alpha) \right\}, \quad \forall \alpha \in (0,\bar{\alpha}]
\end{align*}
Consider the technology $\Acts \eqdef \bigcup_{\alpha \in [0,\bar{\alpha}]} \Acts(\alpha)$. 
\begin{claim}
$\Acts \supseteq \Acts_0$.
\end{claim}
\begin{proof}
For an arbitrary $ (F_0,c_0) \in \Acts_0$, if $\Exlong[F_0]{y} \le e^*(0)$, we have that $(F_0, c_0) \in \Acts(0)$.
Otherwise, let $\Exlong[F_0]{y}=e^*(\alpha_0)$ for some $\alpha_0 \in (0,1)$. We then have
\begin{align*}
\alpha_0 \cdot e^*(\alpha_0)-c^*(\alpha_0) & = \int_0^{\alpha_0} e^*(t) \dd t = \int_0^{\alpha_0} \frac{\underline{u}(\alpha^*)}{- (1-t) \cdot \ln{(1-\alpha^*)}} \dd t = \frac{- \ln(1-\alpha_0) \cdot \underline{u}(\alpha^*)}{- \ln(1-\alpha^*)} \\
& \ge \underline{u}(\alpha_0) = \max_{(F,c)\in \Acts_0} \left( \alpha_0 \cdot \Ex[F]{y} - c \right) \ge \alpha_0 \cdot \Exlong[F_0]{y}-c_0 = \alpha_0 \cdot e^*(\alpha_0) - c_0~,
\end{align*}
where the first inequality follows from the definition of $\alpha^*$. Consequently, $c_0 \ge c^*(\alpha_0)$ and $(F_0, c_0) \in \Acts(\alpha_0)$. 
\end{proof}

Finally, we show that any contract $w$ achieves an expected payoff of at most $\frac{\underline{u}(\alpha^*)}{- \ln (1-\alpha^*)}$ with respect to $\Acts$.
Let $a_w=(F_w,c_w) \in \Acts$ be the action chosen by the agent. Here, we break tie in the best case for the principal. Note that $a_w$ maximizes the agent's utility, i.e.,
\[
\Exlong[F_w]{w(y)} - c_w = \max_{(F,c) \in \Acts} \left( \Exlong[F]{w(y)} - c \right)~.
\]
Note that the maximum utility is achievable since $\Acts$ is constructed as a closed set\footnote{This is the reason that we introduce $\bar{\alpha}$.}.
Suppose $a_w \in \Acts(\alpha_w)$. Then $\Ex[F_w]{y} \le e^*(\alpha_w)$ and $c_w = c^*(\alpha_w)$. 
Consider the following alternatives of possible actions
\[
\left\{ \left( t\cdot F_w+(1-t)\cdot \delta_0,c^*\left((e^*)^{-1} \left(t\cdot e^*(\alpha_w) \right) \right) \right) \right\}_{t \in [0,1]}~.
\]
That is, the agent can potentially play the mixture of $F_w$ and $\delta_0$. For each $t \in [0,1]$, her corresponding utility is then
\[
u_w(t) \eqdef t \cdot \Exlong[F_w]{w(y)} - c^*((e^*)^{-1}(t\cdot e^*(\alpha_w)))~.
\]
Since $t=1$ corresponds to the agent's best action, it must be the case that $u_w'(1) \ge 0$. That is,
\begin{align*}
0 \le u_w'(1) & = \Exlong[F_w]{w(y)} - \left. (c^*)'\left((e^*)^{-1}(t \cdot e^*(\alpha_w))\right) \cdot \left((e^*)^{-1}\right)'(t \cdot e^*(\alpha_w)) \cdot e^*(\alpha_w) \right|_{t=1} \\
& = \Exlong[F_w]{w(y)} - (c^*)'(\alpha_w) \cdot \frac{1}{\left(e^*\right)'(\alpha_w)} \cdot e^*(\alpha_w) \\
& = \Exlong[F_w]{w(y)} - \alpha_w \cdot e^*(\alpha_w)~.
\end{align*}
Here, the first equality follows from the chain rule; the second equality follows from the inverse function rule in calculus; the last equality follows from the definition of $c^*$.
Consequently, the expected payoff of the principal is at most
\[
\Exlong[F_w]{y - w(y)} \le (1-\alpha_w) \cdot e^*(\alpha_w) = \frac{\underline{u}(\alpha^*)}{-\ln(1-\alpha^*)}~.
\]
This upper bound applies to an arbitrary contract $w$ and hence, concludes the proof of the lemma.

\begin{remark*}
We remark that the technology $\Acts$ need not include all possible outcome distributions $F$ as constructed above. The crucial part our proof is to have that the agent can play $t \cdot F + (1-t) \cdot \delta_0$ as long as $F$ is in the technology, so that the step of $u_w'(1) \ge 0$ can go through. To this end, we can as well only consider outcome distributions $\{ \left. t \cdot F_0 + (1-t) \cdot \delta_0 \right| (F_0,c) \in \Acts_0, t \in [0,1] \}$. A similar observation is pointed out by Carroll~\cite{carroll2015robustness} (in Appendix C of his paper) for deterministic contracts.
\end{remark*}

\subsection{Comparison between Optimal Randomized and Deterministic Contracts}
We examine the advantage of randomized contracts over deterministic contracts. We focus on the case when $\Acts_0$ involves only one non-trivial action $(F_0,c_0)$ with $\Ex[F_0]{y} > c_0$.
For notation simplicity, we further normalize $\Ex[F_0]{y}$ to be $1$ and assume $c_0 < 1$.
\begin{theorem*}[Carroll~\cite{carroll2015robustness}]
The optimal deterministic contract achieves a payoff of 
\[
\left( \sqrt{\Ex[F_0]{y}} - \sqrt{c_0}\right)^2 = \left( 1 - \sqrt{c_0} \right)^2~.
\]
\end{theorem*}
Theorem~\ref{thm:main} states that the optimal randomized payoff is 
\[
\max_{\alpha} \frac{\alpha \cdot \Ex[F_0]{y} - c_0}{- \ln (1-\alpha)} = \max_{\alpha} \frac{\alpha - c_0}{- \ln (1-\alpha)}
\]
By taking derivative over $\alpha$, we know that maximum is achieved at $\alpha^*$, the unique solution to 
\[
\alpha^* + \ln(1-\alpha^*) \cdot (1-\alpha^*) = c_0~.
\]
The corresponding payoff is
\[
\frac{\alpha^* - c_0}{- \ln (1-\alpha^*)} = 1-\alpha^*.
\]
Therefore, the ratio between the optimal randomized payoff and  deterministic payoff is
\[
\frac{1-\alpha^*}{\left( 1 - \sqrt{c_0} \right)^2}~.
\]
Notice that the ratio approaches infinity when $c_0 \to 1$. 
\begin{claim}
Let $\alpha(y)$ be the inverse function of $y=x + \ln(1-x) \cdot (1-x)$, where  $x \in [0,1]$, then 
\[
\lim_{y\to 1^-} \frac{1-\alpha(y)}{\left( 1 - \sqrt{y} \right)^2} = +\infty~.
\]
\end{claim}
\begin{proof}
It can be easily verified that $y$ is an increasing function of $x$ and when $x \to 1^{-}$, $y \to 1^{-}$. Therefore we have
\[
\lim_{y\to 1^-} \frac{1-\alpha(y)}{\left( 1 - \sqrt{y} \right)^2} = \lim_{x\to 1^-} \frac{1-x}{\left( 1 - \sqrt{x + \ln(1-x) \cdot (1-x)} \right)^2}=+\infty.
\]
\end{proof}
We conclude that the advantage of randomized contracts over deterministic contracts can be arbitrarily large.
We also plot the ratio as a function of $c_0$ (refer to Figure~\ref{fig:ratio}).
\begin{figure}[H]
    \centering
    \includegraphics[width=.7\textwidth]{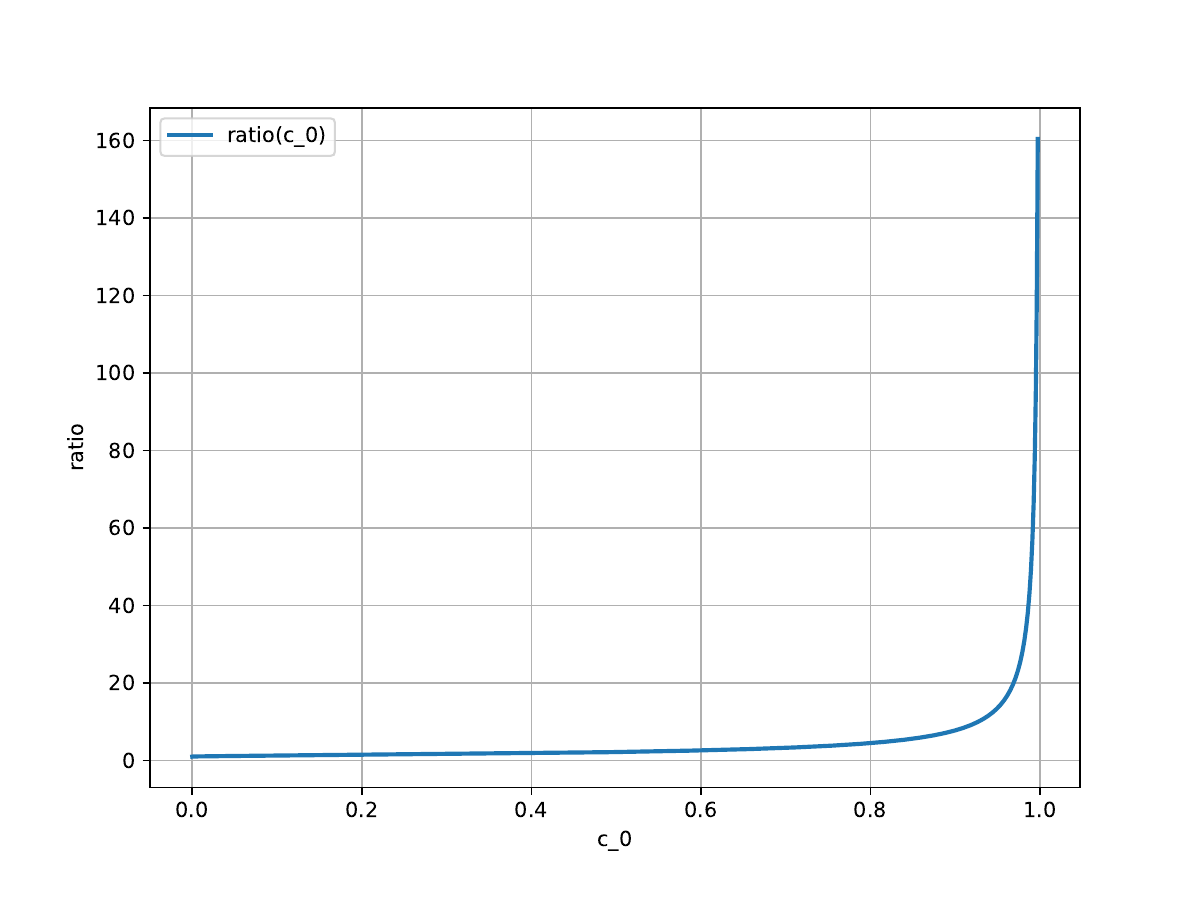}
    \caption{Ratio between the optimal randomized payoff and the optimal deterministic payoff}
    \label{fig:ratio}
\end{figure}

\section{Optimal Randomized Linear Contracts for Teams}
\label{sec:team}
In this section, we generalize our result to the case of contracting with a team comprising multiple agents. The following model is proposed by Dai and Toikka~\cite{dai2022robust}. They generalized the result of Carroll~\cite{carroll2015robustness} and proved that the optimal deterministic contract is linear.

We use the index set $I=\{1,...,n\}$ to denote the agents. A team's production technology is characterized by a tuple $(\Acts,\vect{c},F)$, where $\Acts= \times_{i=1}^n \Acts_i$ denotes the finite action space, $\vect{c}:\Acts \to \mathbb{R}_{+}^n$ represents the cost function of agents, and $F:\Acts \to \Delta
(\mathcal{Y})$ constitutes the family of output distributions. We restrict ourselves to the technologies where the cost incurred by each agent depends solely on their individual action, i.e., $c_i(\vect{a})=c_i(a_i)$ for any unobservable action $\vect{a}=(a_1,...,a_n)$. In addition, we assume that both the observable output set $\mathcal{Y} \subseteq \mathbb{R}_{+}$ and the action space $\Acts$ include finite elements. Similar to the single-agent case, the principal also has the flexibility to employ monetary rewards as incentives to motivate a team based on the outcome. The agents benefit from limited liability protection, requiring payments to be nonnegative.

A contract is defined as a vector-valued function $\vect{w}: \mathcal{Y} \to \mathbb{R}_{+}^n$, articulating the payment rule $\vect{w}(y) = (w_1(y),...,w_n(y))$ for each output $y$. Moreover, a contract $\vect{w}$ is termed linear if there exists some vector $\vect{\alpha}=(\alpha_1,...,\alpha_n) \in [0,1]^n$ such that $w_i(y) = \alpha_i \cdot y$ for all $i \in I$ and $y \in \mathcal{Y}$. For any contract $\vect{w}$, agent $i$'s utility is expressed as $w_i(y) -c_i(a_i)$ and the principal's payoff is given by $y-\sum_{i=1}^n w_i(y)$. The goal of the principal is to design a contract to maximize her payoff. However, the principal is only informed about some technology $(\Acts_0,\vect{c_0},F_0)$ which is part of the true technology $(\Acts,\vect{c},F)$ such that $\Acts
_0 \subseteq \Acts$ and $(\vect{c},F)\arrowvert_{\Acts_0} = (\vect{c_0},F_0)$. That is, the underlying technology contains the known action profile, and the corresponding costs and outcome distributions are consistent. We also assume that $(\Acts_0,\vect{c_0},F_0)$ includes a zero-cost action for the agents. 

Each contract-technology pair $(\vect{w},(\Acts,\vect{c},F))$ induces a normal form game, and a mixed strategy Nash equilibrium always exists due to the finite set $\Acts$. We focus on the equilibrium that maximizes the principal’s payoff. 

From now on,  we fix an arbitrary finite technology set $(\Acts_0,\vect{c}_0,F_0)$. Similar to the case of the single agent case, we define the following auxiliary function
\[
\underline{u}(\vect{\alpha}) \eqdef \max_{\vect{a} \in \Acts_0} \left(\Exlong[F(\vect{a})]{y}-\sum_{i=1}^n \frac{c_i(a_i)}{\alpha_i}\right)~.
\]

\paragraph{Productive known technology.} We assume that there exists $\vect{\alpha}$ with $\sum_{i} \alpha_i \le 1$ and $\underline{u}(\vect{\alpha}) > 0$. This is the same assumption that Dai and Toikka need to guarantee a positive payoff by deterministic contracts.

Let $\vect{\alpha}^*$ denote the critical slope that
\[
\vect{\alpha}^* \in \argmax_{\{\vect{\alpha}\in \mathbb{R}_{+}^n|\sum_{i=1}^n \alpha_i \le 1\}} \left(\frac{\underline{u}(\vect{\alpha})\cdot \sum_{i=1}^n \alpha_i}{-\ln(1-\sum_{i=1}^n \alpha_i)}\right)
\]

We then consider the following cumulative distribution:
\[
\tilde{G}^*(\beta) \eqdef \frac{\ln(1-\beta \cdot \sum_{i=1}^n \alpha_i^*)}{\ln(1-\sum_{i=1}^n \alpha_i^*)}, \quad \forall \beta \in [0,1]~.
\]

\begin{theorem}\label{team}
     The randomized linear contract supporting on the line segment $l(\beta) = \beta \cdot \vect{\alpha}^*$ for $\beta \in [0,1]$, where $\beta$ is drawn according to the cumulative distribution function $\tilde{G}^*(\beta)$,  achieves an expected payoff of $\frac{\underline{u}(\vect{\alpha}^*)\cdot \sum_{i=1}^n \alpha_i^*}{-\ln(1-\sum_{i=1}^n \alpha_i^*)}$.
\end{theorem}
\begin{proof}
For each linear contract $\vect{w}(y)=(w_1(y),...,w_n(y))$, where $w_i(y)=\alpha_i \cdot y$ and a finite technology $(\Acts,\vect{c},F)$, let 
\[
\vect{a}(\vect{\alpha})= (a_1(\vect{\alpha}),...,a_n(\vect{\alpha})) \in \argmax_{\vect{a} \in \Acts} \left(\Exlong[F(\vect{a})]{y}-\sum_{i=1}^n \frac{c_i(a_i)}{\alpha_i}\right).
\]

\begin{claim}[Lemma 4.2 of~\cite{dai2022robust}]
\label{clm:team}
\vect{a}(\vect{\alpha}) is a (pure) Nash equilibrium of the agents with respect to $w_{\vect{\alpha}}$.
\end{claim}
We shall consider $\vect{a}(\vect{\alpha})$ as the actions of the agents for every linear contract.
Furthermore, let
\[
e(\vect{\alpha}) \eqdef \Exlong[F(\vect{a}(\vect{\alpha}))]{y} \quad \text{,} \quad \vect{c}(\vect{\alpha}) = (c_1(a_1(\vect{\alpha})),...,c_n(a_n(\vect{\alpha})))
\]
and
\[
u(\vect{\alpha}) \eqdef \max_{\vect{a} \in \Acts} \left( \Exlong[F(\vect{a})]{y}-\sum_{i=1}^n \frac{c_i(a_i)}{\alpha_i}\right).
\]

Now, we establish a convex program to lower bound the principal's expected payoff of any randomized linear contract.
\begin{claim}
For an arbitrary randomized linear contract with cumulative distribution function $G$ on the slopes $\vect{\alpha}$, the expected payoff of the principal is lower bounded by the following program.
\begin{align*}
\label{lp3}
\min_{u}: \quad & \int_{\{\vect{\alpha} \in \mathbb{R}_{+}^n|\sum_{i=1}^n \alpha_i \le 1\}} \left(1-\sum_{i=1}^n \alpha_i\right) \cdot \left(u(\vect{\alpha})+\sum_{i=1}^n \frac{\partial u(\vect{\alpha})}{\partial \alpha_i} \cdot \alpha_i\right) \, \dd G(\vect{\alpha})  \tag{P3}\\
{\normalfont \text{subject to}}: \quad & u(\vect{\alpha}) \ge \underline{u}(\vect{\alpha}), \hspace{4cm} \forall \vect{\alpha} \in  \left\{\vect{\alpha} \in \mathbb{R}_{+}^n \middle| \sum_{i=1}^n \alpha_i \le 1\right\}
\end{align*}
\end{claim}

\begin{proof}
Given any $\vect{\alpha}=(\alpha_1,...,\alpha_n)$, consider the linear contract $\vect{w}(y)=(w_1(y),...,w_n(y))$, where $w_i(y)=\alpha_i \cdot y$. 
By the definition of $e(\vect{\alpha})$, $\vect{c}(\vect{\alpha})$ and $u(\vect{\alpha})$, we have
\[
u(\vect{\alpha})=e(\vect{\alpha}) -\sum_{i=1}^n \frac{c_i(\vect{\alpha})}{\alpha_i}
\]
and thus
\[
\frac{\partial u(\vect{\alpha})}{\partial \alpha_i} = \frac{\partial e(\vect{\alpha})}{\partial \alpha_i} - \sum_{j=1}^n \frac{1}{\alpha_j} \cdot \frac{\partial c_j(\vect{\alpha})}{\partial \alpha_i} + \frac{c_i(\vect{\alpha})}{\alpha_i^2} = \frac{c_i(\vect{\alpha})}{\alpha_i^2},
\]
where the second equality holds because $e(\alpha_i, \vect{\alpha}_{\text{-}i}),c_j(\alpha_i, \vect{\alpha}_{\text{-}i})$ as an one-dimension function of $\alpha_i$ is piece-wise constant. Consequently, we have $e(\vect{\alpha})=u(\vect{\alpha})+\sum_{i=1}^n \frac{\partial u(\vect{\alpha})}{\partial \alpha_i} \cdot \alpha_i$. 

Therefore, the principal's payoff corresponding to $\vect{\alpha}$ is 
     $$\left(1-\sum_{i=1}^n \alpha_i\right) \cdot \left(u(\vect{\alpha})+\sum_{i=1}^n \frac{\partial u(\vect{\alpha})}{\partial \alpha_i} \cdot  \alpha_i\right).$$
     And the constraint holds trivially by $\Acts_0 \subseteq \Acts$. We conclude our proof.
\end{proof}

Now, we are left to solve program \eqref{lp3}.

\paragraph{Degenerated case: $\vect{\alpha}^*=\vect{0}$.} 
We first study the degenerated case when $\vect{\alpha}^*=\vect{0}$. Notice that our distribution $\tilde{G}$ degenerates to a point mass at $\beta=0$. It suffices to verify that the stated payoff 
\[
\frac{\underline{u}(\vect{\alpha}^*)\cdot \sum_{i=1}^n \alpha_i^*}{-\ln(1-\sum_{i=1}^n \alpha_i^*)}
\]
coincides with the following optimal deterministic payoff Dai and Toikka~\cite{dai2022robust}:
\[
\max_{\vect{a} \in \Acts_0} \left( \sqrt{\Exlong[F(\vect{a})]{y}} -\sum_{i=1}^n \sqrt{c_i(a_i)}\right)^2~.
\]
First, there must exist an action $\vect{a}^* \in \Acts_0$ satisfying $c_i(a_i^*)=0$ for every $i \in I$, and
\begin{equation}
\label{eqn:team1}
\Exlong[F(\vect{a}^*)]{y} \ge \frac{\underline{u}(\vect{\alpha})\cdot \sum_{i=1}^n \alpha_i}{-\ln(1-\sum_{i=1}^n \alpha_i)}, \forall \vect{\alpha} \in  \left\{\vect{\alpha} \in \mathbb{R}_{+}^n \middle| \sum_{i=1}^n \alpha_i \le 1\right\}
\end{equation}
Next, we prove that
\begin{equation}
\label{eqn:team_degenerate}
\Exlong[F(\vect{a}^*)]{y} \ge \left( \sqrt{\Exlong[F(\vect{a})]{y}} -\sum_{i=1}^n \sqrt{c_i(a_i)}\right)^2, \forall \vect{a} \in \Acts_0.
\end{equation}
Fix an arbitrary $\vect{a} \in \Acts_0$. If $\Ex[F(\vect{a}^{*})]{y} \ge \Ex[F(\vect{a})]{y}$, we have $ \sqrt{\Ex[F(\vect{a}^{*})]{y}} \ge \sqrt{\Ex[F(\vect{a})]{y}} \ge \sqrt{\Ex[F(\vect{a})]{y}}-\sum_{i=1}^n \sqrt{c_i(a_i)}$.
Else, by Cauchy-Schwarz Inequality, we have
\begin{equation}
\label{eqn:team2}
   \sum_{i=1}^n \frac{c_i(a_i)}{\alpha_i} \ge \frac{(\sum_{i=1}^n \sqrt{c_i(a_i)})^2}{\sum_{i=1}^n \alpha_i} 
\end{equation}
let $h(\vect{\alpha}) \eqdef \Ex[F(\vect{a}^*)]{y} \cdot (- \ln (1-\sum_{i=1}^n \alpha_i)) - \Ex[F(\vect{a})]{y} \cdot (\sum_{i=1}^n \alpha_i) + (\sum_{i=1}^n \sqrt{c_i(a_i)})^2$. By combining \eqref{eqn:team1} with \eqref{eqn:team2}, we have $h(\vect{\alpha}) \ge 0$ for every $\vect{\alpha} \in [0,1]$.
Note that $h$ attains its minimum when $\sum_{i=1}^n \alpha_i=1-\frac{\Ex[F(\vect{a}^{*})]{y}}{\Ex[F(\vect{a})]{y}}$. Thus, we have
\begin{multline*}
  \Exlong[F(\vect{a}^{*})]{y} \cdot \left( -\ln\left( \frac{\Ex[F(\vect{a}^{*})]{y}}{\Ex[F(\vect{a})]{y}}\right)\right) -\Exlong[F(\vect{a})]{y} +\Exlong[F(\vect{a}^{*})]{y} +(\sum_{i=1}^n \sqrt{c_i(a_i)})^2 \ge 0 \\
\Longrightarrow 
\sqrt{\Exlong[F(\vect{a}))]{y}}-\sum_{i=1}^n \sqrt{c_i(a_i)} \le \sqrt{\Exlong[F(\vect{a})]{y}} - \sqrt{\Exlong[F(\vect{a})]{y} -\Exlong[F(\vect{a}^{*})]{y} +\Exlong[F(\vect{a}^{*})]{y} \cdot \ln\left( \frac{\Ex[F(\vect{a}^{*})]{y}}{\Ex[F(\vect{a})]{y}}\right)} \\
= \sqrt{\Exlong[F(\vect{a})]{y}} \cdot \left(1-\sqrt{1-\frac{\Ex[F(\vect{a}^{*})]{y}}{\Ex[F(\vect{a})]{y}}+\frac{\Ex[F(\vect{a}^{*})]{y}}{\Ex[F(\vect{a})]{y}} \cdot \ln\left(\frac{\Ex[F(\vect{a}^{*})]{y}}{\Ex[F(\vect{a})]{y}}\right)}\right) \\
\le \sqrt{\Exlong[F(\vect{a})]{y}} \cdot \sqrt{\frac{\Ex[F(\vect{a}^{*})]{y}}{\Ex[F(\vect{a})]{y}}} = \sqrt{\Exlong[F(\vect{a}^*)]{y}},
\end{multline*}
where the last inequality follows from the mathematical fact that $1-\sqrt{1-x+x\cdot \ln x} \le \sqrt{x}$ for $x \in [0,1]$. This finishes the proof of equation~\eqref{eqn:team_degenerate}.
In other words, $\vect{a}^*$ is the maximizer of 
\[
\max_{\vect{a} \in \Acts_0} \left( \sqrt{\Exlong[F(\vect{a})]{y}} -\sum_{i=1}^n \sqrt{c_i(a_i)}\right)^2~.
\]

\paragraph{General case: $\vect{\alpha}^* \ne \vect{0}$.}
According to our construction, the objective value of \eqref{lp3} is
\begin{align*}
& \phantom{=} \int_0^1 \left(1-\beta \cdot \sum_{i=1}^n \alpha_i^*\right) \cdot \left(u(\beta \cdot \vect{\alpha}^*)+\sum_{i=1}^n \frac{\partial u(\beta \cdot \vect{\alpha}^*)}{\partial \alpha_i} \cdot \beta \cdot \alpha_i^*\right) \dd \tilde{G}^*(\beta)\\
& = \int_0^1 \left(1-\beta \cdot \sum_{i=1}^n \alpha_i^*\right) \cdot \left(u(\beta \cdot \vect{\alpha}^*)+\sum_{i=1}^n \frac{\partial u(\beta \cdot \vect{\alpha}^*)}{\partial \alpha_i} \cdot \beta \cdot \alpha_i^*\right) \frac{- \sum_{i=1}^n \alpha_i^*}{\ln(1-\sum_{i=1}^n \alpha_i^*) \cdot (1-\beta \cdot \sum_{i=1}^n \alpha_i^*)} \dd \beta\\
& = \frac{-  \sum_{i=1}^n \alpha_i^*}{\ln(1-\sum_{i=1}^n \alpha_i^*)}  \cdot \int_0^1 \left(\beta \cdot u(\beta \cdot \vect{\alpha}^*)\right)' \dd \beta =  \frac{ (\sum_{i=1}^n \alpha_i^*) \cdot u(\vect{\alpha}^*)}{-\ln(1-\sum_{i=1}^n \alpha_i^*)} \ge \frac{ (\sum_{i=1}^n \alpha_i^*) \cdot \underline{u}(\vect{\alpha}^*)}{-\ln(1-\sum_{i=1}^n \alpha_i^*)},
\end{align*}
where the inequality follows from the constraint, which concludes the proof of the theorem.
\end{proof}

\begin{remark*}
We would like to leave a few comments regarding the above result.
\begin{itemize}
\item It is crucial for us to formalize the program in terms of the ``utility'' function $u(\vect{\alpha})$. It can be interpreted as the utility of the whole team that uniquely determines the costs of each agent and the expected outcome. This approach has been applied in the literature of mechanism design, e.g., the multiple-good monopoly problem~\cite{daskalakis2017strong}. Similar to the single agent setting, we also treat the technology $\Acts$ as a menu for the agents.

\item Our theorem for the team setting is not as strong as Theorem~\ref{thm:main} for the single agent setting. Indeed, we only provide a lower bound on the optimal randomized payoff and is not able to provide a matching upper bound. The technical difficulty is to characterize all possible Nash equilibrium of the induced game among the agents that prevents us to generalize the proof of Lemma~\ref{lem:upper} for the single agent setting. Indeed, Claim~\ref{clm:team} only provides one specific Nash equilibrium of $(F(\vect{\alpha}), \vect{c}(\vect{\alpha}))$ with respect to the linear contract $\vect{w}_{\vect{\alpha}}$. 

On the other hand, our randomized contract is optimal within the current analysis. Specifically, our construction is optimal with respect to \eqref{lp3}. Consider the following function
\[
u^*(\vect{\alpha}) = \left(\frac{\underline{u}(\vect{\alpha}^*)\cdot \sum_{i=1}^n \alpha_i^*}{-\ln(1-\sum_{i=1}^n \alpha_i^*)}\right) \cdot \frac{-\ln(1-\sum_{i=1}^n \alpha_i)}{\sum_{i=1}^n \alpha_i}
\]
that satisfies the constraint $u^*(\vect{\alpha}) \ge \underline{u}(\vect{\alpha})$ for every $\vect{\alpha} \in \{\vect{\alpha} \in \mathbb{R}_{+}^n|\sum_{i=1}^n \alpha_i \le 1\}$ according to the definition of $\vect{\alpha}^*$.
Moreover, the payoff for every $\vect{\alpha}$ is 
\[
\left(1-\sum_{i=1}^n \alpha_i\right) \cdot \left(u^*(\vect{\alpha})+\sum_{i=1}^n \frac{\partial u^*(\vect{\alpha})}{\partial \alpha_i} \cdot \alpha_i\right) \equiv \frac{\underline{u}(\vect{\alpha}^*)\cdot \sum_{i=1}^n \alpha_i^*}{-\ln(1-\sum_{i=1}^n \alpha_i^*)}~.
\]
Therefore, \eqref{lp3} is at most $\frac{\underline{u}(\vect{\alpha}^*)\cdot \sum_{i=1}^n \alpha_i^*}{-\ln(1-\sum_{i=1}^n \alpha_i^*)}$ for any distribution $G$.

We conjecture that our randomized linear contract is optimal among all randomized contracts and leave it for future work.
\end{itemize}
\end{remark*} 
\bibliographystyle{plain}
\bibliography{contract}

\newpage
\appendix

\section{Deterministic Contracts}
\label{app:deterministic}

 Carroll~\cite{carroll2015robustness} demonstrates the optimality of a linear contract for the principal. We use a optimization-based approach to give an alternative proof of this result. Furthermore, this method seamlessly extends to the case with multi-observable outcomes. 

In this section, we abuse the notation $V_{P}(\Acts_0)$ to denote the principal's worst-case payoff by using a deterministic contract:
\[
\max_{w} \min_{\Acts \supseteq \Acts_0} \left( \Exlong[y \sim F(w)]{y-w(y)} \right),
\]
where $(F(w),c(w)) \in \Acts$ is the utility-maximizing action of the agent with respect to contract $w$.

\paragraph{Tie-breaking.} We assume that the agent break ties in the \emph{worst case} but make an extra assumption on the known technology $\Acts_0$. We assume that the following value of $\lambda^*$ is finite: 
\[
\lambda^* \in \argmax_{\lambda} \max_{(F_0,c_0) \in \Acts_0} \left( \frac{\lambda}{\lambda+1}\cdot \Exlong[F_0]{y} - \lambda \cdot c_{0} \right)~.
\]
We remark that for the corner case when $\lambda^* = +\infty$, the optimal contract degenerates to $w^*(y) = 0, \forall y$ for which we necessarily need a best-case tie-breaking rule. We hence omit this technical detail by making the assumption in the following analysis.

\begin{lemma}
    For the single outcome case, given any deterministic contract $w$, the expected payoff for the principal can be captured by the following program.
\begin{align*}
 \inf_{(F,c)}: \quad & \Exlong[F]{y-w(y)} \\
\text{subject to:} \quad & \Exlong[F]{w(y)} -c \ge \max_{(F_0,c_0) \in \Acts_0} \left( \Exlong[F_0]{w(y)} - c_0 \right) \\
& c \ge 0
\end{align*}
\end{lemma}

\begin{proof}
   The first constraint arises from considerations of individual
rationality and the requirement that the action cost $c \ge 0$ is natural. When the agent chooses the action $(F,c)$, the objective value $\Exlong[F]{y-w(y)}$ represents the assured payoff for the principal. Therefore, we can formulate worst-case analysis as the above minimization problem.
\end{proof}

\begin{theorem}[Theorem 1 of \cite{carroll2015robustness}]
    For the single outcome case, the optimal deterministic contract is linear.
\end{theorem}

\begin{proof}
Given any contract $w$, we consider the Lagrangian dual of the above program.
\begin{align*}
\mathcal{L}_{w}(F,\lambda) & = \Exlong[F]{y-w(y)} + \lambda \cdot \left( \max_{(F_0,c_0) \in \Acts_0} \left( \Exlong[F_0]{w(y)} - c_0 \right) - \Exlong[F]{w(y)} \right) \\
& = \lambda \cdot  \max_{(F_0,c_0) \in \Acts_0} \left( \Exlong[F_0]{w(y)} - c_0 \right) + \Exlong[F]{y - (\lambda+1)\cdot w(y))}
\end{align*}
Since $\mathcal{L}(F,\lambda)$ is linear in both $F$ and $\lambda$, we have that
\begin{multline*}
\inf_{F} \sup_{\lambda} \mathcal{L}_{w}(F,\lambda) = \sup_{\lambda} \inf_{F} \mathcal{L}_{w}(F,\lambda) \le \sup_{\lambda} \mathcal{L}_{w}\left( \frac{\lambda}{\lambda+1}\cdot F_1 + \frac{1}{\lambda+1} \cdot \delta_0, \lambda \right) \\
= \sup_{\lambda}\left(\frac{\lambda}{\lambda+1}\cdot \Exlong[F_1]{y} - \lambda \cdot c_1 -w(0) \right)\le \sup_{\lambda} \max_{(F_0,c_0) \in \Acts_0} \left( \frac{\lambda}{\lambda+1}\cdot \Exlong[F_0]{y} - \lambda \cdot c_{0} \right),
\end{multline*}
where $(F_1,c_1) = \argmax\limits_{(F_0,c_0)\in \Acts_0} \left(\Exlong[F_0]{w(y)} - c_0 \right)$.
Consequently, we have
\begin{align*}
V_P(\Acts_0) \le \sup_{\lambda} \max_{(F_0,c_0) \in \Acts_0} \left( \frac{\lambda}{\lambda+1}\cdot \Exlong[F_0]{y} - \lambda \cdot c_{0} \right) = \max_{(F_0,c_0) \in \Acts_0} \max_{\lambda} \left( \frac{\lambda}{\lambda+1}\cdot \Exlong[F_0]{y} - \lambda \cdot c_{0} \right) \\
= \max_{(F_0,c_0) \in \Acts_0} \left( \sqrt{\Exlong[F_0]{y}} - \sqrt{c_0}\right)^2
\end{align*}
Let $(F^*,c^*) \in \Acts_0$ to be the maximizer of the above quantity. 

On the other hand, the upper bound is achievable by a linear contract. Specifically, define $\lambda^* \eqdef \sqrt{\frac{\Ex[F^*]{y}}{c^*}} - 1$ and $w^*(y) \eqdef \frac{y}{\lambda^* + 1} = \sqrt{\frac{c^*}{\Ex[F^*]{y}}} \cdot y$. According to our assumption, $\lambda^* < +\infty$.
Then we have the following.
\begin{multline*}
V_P(\Acts_0) \ge \inf_{F} \mathcal{L}_{w^*}(F,\lambda^*) =\inf_{F} \left(\lambda^* \cdot  \max_{(F_0,c_0) \in \Acts_0} \left( \frac{\Ex[F_0]{y}}{\lambda^*+1} - c_0 \right) + \Exlong[F]{y - (\lambda^*+1)\cdot \frac{y}{\lambda^*+1}}\right) \\
= \lambda^* \cdot  \max_{(F_0,c_0) \in \Acts_0} \left( \frac{\Ex[F_0]{y}}{\lambda^*+1} - c_0 \right) \ge \lambda^*\cdot \left( \frac{\Ex[F^*]{y}}{\lambda^*+1} - c^* \right) = \left( \sqrt{\Exlong[F^*]{y}} - \sqrt{c^*}\right)^2,
\end{multline*}
which concludes that the linear contract $w^*$ is optimal.

\end{proof}

\subsection{Extensions to multi-observable outcomes}
\paragraph{Multi-observable outcomes.} In this setting, we assume the observable outcome is a $k$-dimension vector $\vect{y} \in \mathcal{Y} \subseteq \mathbb{R}^k$, where the first entry $y_1$ is the monetary reward the principal can receive. A contract is a multivariate function $w: \mathcal{Y} \to \mathbb{R}$ and the principal's payoff is $y_1-w(\vect{y})$. In addition, we are given a convex function $b:\mathbb{R}^k \to \mathbb{R}$ as a lower bound on the cost of each action. Specifically, the technology $\Acts$ satisfies that for every action $(F,c) \in \Acts$, $c \ge b\left(\Exlong[F]{\vect{y}}\right)$. Naturally, the technology $\Acts_0$ known to the principal adheres to this constraint.

Similar to the previous analysis, we use $V_P(\Acts_0)$ to denote the principal's worst-case expected payoff
\[
\max_{w} \min_{\Acts \supseteq \Acts_0} \left( \Exlong[\vect{y} \sim F(w)]{y_1-w(\vect{y})} \right)~.
\]

\paragraph{Tie-breaking.} We also assume that the agent break ties in the \emph{worst case} but make an extra assumption on the known technology $\Acts_0$. We assume that the following value of $\lambda^*$ is finite: 
\[
\lambda^* \in \argmax_{\lambda} \max_{(F_0,c_0) \in \Acts_0} \left( \frac{\lambda}{\lambda+1}\cdot \Exlong[F_0]{y_1} - \lambda \cdot c_0+\lambda \cdot b\left(\frac{\lambda}{\lambda+1}\Exlong[F_0]{\vect{y}}\right) \right)~.
\]
Similarly, for the corner case when $\lambda^* = +\infty$, the optimal contract degenerates and we have to use a best-case tie-breaking rule.

We also formulate the following optimization problem to express the principal's expected payoff. 

\begin{lemma}
For the multi-observable outcomes, given any deterministic contract $w$, the expected payoff for the principal can be captured by the following program.
\begin{align*}
 \inf_{(F,c)}: \quad & \Exlong[F]{y_1-w(\vect{y})} \\
\text{subject to:} \quad & \Exlong[F]{w(\vect{y})} -c \ge \max_{(F_0,c_0) \in \Acts_0} \left( \Exlong[F_0]{w(\vect{y})}-c_0  \right) \\
& c \ge b\left(\Exlong[F]{\vect{y}}\right)
\end{align*}

\end{lemma}

\begin{proof}
The constraints mirror those of the single-outcome case. The only modification required is to update the objective value to $\Exlong[F]{y_1-w(\vect{y})}$, thereby completing the construction of this program.
\end{proof}

\begin{theorem}[Theorem 2 of \cite{carroll2015robustness}]
    For the multi-observable outcomes, the optimal deterministic contract is linear.
\end{theorem}

\begin{proof}

Given any contract $w$, consider the Lagrangian dual of the above program.
\begin{align*}
\mathcal{L}_{w}(F,\lambda) & = \Exlong[F]{y_1-w(\vect{y})} + \lambda \cdot \left( \max_{(F_0,c_0) \in \Acts_0} \left( \Exlong[F_0]{w(\vect{y}}-c_0  \right) - \Exlong[F]{w(\vect{y})}+b\left(\Exlong[F]{\vect{y}}\right) \right) \\
& = \lambda \cdot  \max_{(F_0,c_0) \in \Acts_0} \left( \Exlong[F_0]{w(\vect{y}}-c_0\right) + \Exlong[F]{y_1 - (\lambda+1)\cdot w(\vect{y}))}+\lambda \cdot b\left(\Exlong[F]{\vect{y}}\right)
\end{align*}
Since $\mathcal{L}(F,\lambda)$ is linear in both $F$ and $\lambda$, we have that
\begin{multline*}
\inf_{F} \sup_{\lambda} \mathcal{L}_{w}(F,\lambda) = \sup_{\lambda} \inf_{F} \mathcal{L}_{w}(F,\lambda) \le \sup_{\lambda} \mathcal{L}_{w}\left( \frac{\lambda}{\lambda+1}\cdot F_1 + \frac{1}{\lambda+1} \cdot \delta_0, \lambda \right) \\
= \sup_{\lambda} \left( \frac{\lambda}{\lambda+1}\cdot \Exlong[F_1]{y_1} - \lambda \cdot c_1+\lambda \cdot b\left(\frac{\lambda}{\lambda+1}\Exlong[F_1]{\vect{y}}\right) -w(0) \right) \\
\le \sup_{\lambda} \max_{(F_0,c_0) \in \Acts_0} \left( \frac{\lambda}{\lambda+1}\cdot \Exlong[F_0]{y_1} - \lambda \cdot c_0+\lambda \cdot b\left(\frac{\lambda}{\lambda+1}\Exlong[F_0]{\vect{y}}\right) \right),
\end{multline*}
where $(F_1,c_1) = \argmax\limits_{(F_0,c_0)\in \Acts_0} \left( \Exlong[F_0]{w(\vect{y}}-c_0\right)$.

Consequently, for any contract $w$, there holds 
\begin{align}
\label{upbound}
V_P(\Acts_0) & \le \sup_{\lambda} \max_{(F_0,c_0) \in \Acts_0} \left( \frac{\lambda}{\lambda+1}\cdot \Exlong[F_0]{y_1} - \lambda \cdot c_0+\lambda \cdot b\left(\frac{\lambda}{\lambda+1}\Exlong[F_0]{\vect{y}}\right) \right).
\end{align}
According to our assumption, the supreme is achieved at a finite $\lambda^*$. Let $(F^*,c^*) \in \Acts_0$ be the corresponding maximizer of the inner maximization.

We claim the upper bound can be achieved by a linear contract. Specifically, denote $p(\vect{y})= \sum_{i=1}^k p_i y_i +\beta$ as  the tangent plane of $b(\cdot)$ at $\frac{\lambda^*}{\lambda^*+1}\Exlong[F^*]{\vect{y}}$. Then, we have
\[
b(\vect{y}) \ge p(\vect{y}) = \sum_{i=1}^k p_i y_i +\beta
\]
and
\[
b\left(\frac{\lambda^*}{\lambda^*+1}\Exlong[F^*]{\vect{y}}\right)= p\left(\frac{\lambda^*}{\lambda^*+1}\Exlong[F^*]{\vect{y}}\right) = \frac{\lambda^*}{\lambda^*+1} \cdot \sum_{i=1}^k p_i \Exlong[F^*]{y_i}+\beta.
\]
Consider the linear contract $w^*(\vect{y}) \eqdef \frac{1}{\lambda^*+1} \cdot \left( y_1+\lambda^* \cdot \sum_{i=1}^k p_i y_i\right)$. Then we have the following.
\begin{multline*}
V_{P}(\Acts_0) \ge \inf_{F} \mathcal{L}_{w^*}(F,\lambda^*) \\
= \inf_{F} \left( \lambda^* \cdot  \max_{(F_0,c_0) \in \Acts_0} \left( \Exlong[F_0]{w^*(\vect{y})}-c_0\right) + \Exlong[F]{y_1 - (\lambda^*+1)\cdot w^*(\vect{y})}+\lambda^* \cdot b\left(\Exlong[F]{\vect{y}}\right) \right) \\
\ge \inf_{F} \left( \lambda^* \cdot   \left( \Exlong[F^*]{w^*(\vect{y})}-c^*\right) - \lambda^* \cdot \Exlong[F]{ \sum_{i=1}^k p_i y_i}+\lambda^* \cdot \left( \sum_{i=1}^k p_i \Exlong[F]{y_i} + \beta\right) \right)\\
=\lambda^* \cdot   \left( \frac{1}{\lambda^*+1} \Exlong[F^*]{y_1} +\frac{\lambda^*}{\lambda^*+1} \sum_{i=1}^k p_i \Exlong[F^*]{y_i} - c^*+\beta \right)\\
= \frac{\lambda^*}{\lambda^*+1}\cdot \Exlong[F^*]{y_1} - \lambda^* \cdot c^*+\lambda^* \cdot b\left(\frac{\lambda^*}{\lambda^*+1}\Exlong[F^*]{\vect{y}}\right),
\end{multline*}
which concludes that the linear contract $w^*$ is optimal.

\end{proof} 
\end{document}